\pdfoutput=1
\documentclass[11pt,a4paper]{article}
\usepackage{lmodern}
\usepackage[T1]{fontenc}
\usepackage{fullpage}
\usepackage{float}
\floatstyle{boxed}
\restylefloat{figure}
\usepackage[tbtags]{amsmath}
\allowdisplaybreaks[4]
\usepackage{amssymb}
\usepackage{amsthm}
\usepackage{hyperref}
\hypersetup{colorlinks={true},linkcolor={blue},citecolor=red}
\usepackage[capitalise,nameinlink]{cleveref}
\usepackage[numbers,sort&compress]{natbib}
\usepackage[font=footnotesize]{caption}

\newtheorem{theorem}{Theorem}

\newtheorem{claim}{Claim}
\theoremstyle{definition}

\theoremstyle{remark}
\newtheorem*{fact}{Fact}
\crefname{claim}{Claim}{Claims}

\def \R{\mathbb R}
\newcommand{\sset}[1]{\left\{ #1\right\}}
\newcommand{\ssets}[1]{\{ #1\}}
\DeclareMathOperator*{\expect}{\mathbb E}
\newcommand{\vecc}[1]{\ensuremath{\mathbf{#1}}}
\newcommand{\probability}[1]{\ensuremath{\mathrm{Pr}\left[#1\right]}}
\newcommand{\bmin}{\hat{t}_{\min}}

\newcommand{\bids}{\hat{\vecc{t}}}
\newcommand{\bidi}{\hat{t}_i}
\newcommand{\bidk}{\hat{t}_k}
\newcommand{\bid}[1]{\ensuremath{\hat{t}_{#1}}}
\newcommand{\lrg}{L}

\DeclareMathOperator*{\argmin}{argmin}
\DeclareMathOperator*{\second}{sec}
\newcommand{\bsec}{\hat{t}_{\second}}

\renewcommand{\paragraph}[1]{\medskip\noindent\textbf{#1.\;}}

\title{The Anarchy of Scheduling Without Money\thanks{Supported by ERC Advanced
Grant 321171 (ALGAME), EPSRC grant EP/M008118/1 and the Alexander von Humboldt Foundation with funds from the German Federal Ministry of Education and Research (BMBF).
\newline 
A preliminary
version of this paper, not including all results, appeared in SAGT'16~\citep{GKK16}. 
\newline
A significant part of this work was carried out while Yiannis Giannakopoulos was at the University of Liverpool and Maria Kyropoulou was at the University of Oxford.}}

\author{Yiannis Giannakopoulos\thanks{Department of Mathematics, TU Munich. Email: \href{mailto:yiannis.giannakopoulos@tum.de}{\nolinkurl{yiannis.giannakopoulos@tum.de}}   }
\and Elias Koutsoupias\thanks{Department of Computer Science, University of Oxford. Email: \href{mailto:elias@cs.ox.ac.uk}{\nolinkurl{elias@cs.ox.ac.uk}}    }
\and Maria Kyropoulou\thanks{School of Computer Science and Electronic Engineering, University of Essex. Email: \href{mailto:maria.kyropoulou@essex.ac.uk}{\nolinkurl{maria.kyropoulou@essex.ac.uk}}    }}

\date{December 4, 2018}
\begin{document}
\maketitle

\begin{abstract}
We consider the scheduling problem on $n$ strategic unrelated machines when no
payments are allowed, under the objective of minimizing the makespan. We adopt
the model introduced in [Koutsoupias 2014] where a machine is \emph{bound} by
her declarations in the sense that if she is assigned a particular job then she
will have to execute it for an amount of time at least equal to the one she
reported, even if her private, true processing capabilities are actually faster.
We provide a (non-truthful) randomized algorithm whose pure \emph{Price of
Anarchy} is arbitrarily close to $1$ for the case of a single task and close to
$n$ if it is applied independently to schedule many tasks, which is
asymptotically optimal for
the natural class of anonymous, task-independent algorithms.
Previous work considers the constraint of truthfulness and proves a tight approximation ratio of $(n+1)/2$ for one task which generalizes to $n(n+1)/2$ for many tasks.
Furthermore, we revisit the truthfulness case and reduce the latter approximation ratio for many tasks down to $n$, asymptotically matching the best known lower bound. This is done via a detour to the relaxed, fractional version of the problem, for which we are also able to provide an optimal approximation ratio of $1$. Finally, we mention that all our algorithms achieve optimal ratios of $1$ for the social welfare objective.
\end{abstract}

\section{Introduction}
We consider a variant of the scheduling problem proposed by \citet{K14} where no
payments are allowed and the machines are bound by their declarations. In
particular, the goal is to allocate a set of tasks to strategic unrelated machines while minimizing the makespan. The time/cost needed by a machine to execute a task is private information of the machine. Each machine is rational and selfish, and will misreport her costs in an attempt to minimize her own overall running time, under the assumption that if she is allocated a task, she will execute it for at least the declared cost (more specifically, for the maximum among her true and reported execution times). We are interested in designing allocation protocols that do not use payments and the stable outcomes are not far from the non-strategic, centrally enforced optimum makespan.

The field of Mechanism Design \cite{Nisan:2007:AGT:1296179} focuses on the
implementation of desired outcomes. Given the strategic behaviour of the players
who provide the input and a specific objective function that measures the
quality of the outcome, the challenge is to design mechanisms which are able to elicit a desired behaviour from the players, while at the same time optimizing that objective value. A primary designer goal that has been extensively studied is that of \emph{truthfulness}, under the central solution concept of dominant strategies: a player should be able to optimize her own individual utility by reporting truthfully, no matter what strategies the other players follow. However, achieving this is not always compatible with maintaining a good objective value \citep{10.2307/1914083,SATTERTHWAITE1975187}. The introduction of \emph{payments} was suggested as a means towards achieving these goals, since a carefully designed payment scheme incentivizes the players to make truthful declarations. The goal now becomes to design such algorithms (termed \emph{mechanisms}) which utilize monetary compensations in order to impose truthful behaviour while optimizing the objective function.
A prominent positive result exists for utilitarian settings at which the objective function is the social welfare, where the well-known VCG mechanism \citep{Vickrey61,Clarke71,Groves73}  constitutes such an optimal mechanism. The study of the algorithmic aspect of mechanism design was initiated by \citet{Nisan:2001aa}, and since then a significant body of work has been dedicated to optimization problems from the mechanism design point of view (see e.g.\ \citep{AT01,DDDR11,Lavi:2011:TNM:2049697.2049699}).

There are many situations, though, where the use of payments might be considered unethical \cite{Nisan:2007:AGT:1296179}, illegal (e.g.\ organ donations) or even just impractical. For this reason researchers have started turning their attention to possible ways of achieving truthfulness without the use of payments. In such a setting, in order to circumvent Social Choice impossibility results (e.g.\ the seminal Gibbard-Satterthwaite \citep{10.2307/1914083,SATTERTHWAITE1975187} theorem), domains with richer structure have to be considered. \citet{Procaccia:2009:AMD:1566374.1566401} were the first to consider achieving truthfulness without using payments, by sacrificing the optimality of the solution and settling for just an approximation, in the context of facility location problems.
This work was extended in \citep{10.2307/40800845,Lu:2010:AOS:1807342.1807393}.
Similar questions have been considered in the context of inter-domain routing \citep{Levin:2008:IRG:1374376.1374388}, in assignment problems \citep{Dughmi:2010:TAW:1807342.1807394}, and in the setting of allocating items to two players (with the use of a certain artificial currency) \citep{Guo:2010:SAM:1838206.1838324}. Moreover, exact (as opposed to approximate) mechanism design without money has a rich history in the social choice literature, e.g.\ \cite{Moulin1980} characterizes functions that are truthfully implementable  when the preferences of the agents are single-peaked.

Clearly, truthfulness is a property desired by every mechanism designer; if the mechanism can ensure that no player can benefit from misreporting, the  designer knows what kind of player behaviour and outcome to expect. Moreover, the focus on truthful mechanisms has been largely motivated by the Revelation Principle stating that essentially every equilibrium state of a mechanism can be simulated by a truthful mechanism which achieves the same objective. However this is no longer possible in the variant we examine here, due to the fact that the players are bound by their declarations and thus do not have quasi-linear utilities. So, it is no longer without loss of generality if we restrict attention to truthful mechanisms. For mechanisms that are not truthful, \emph{Price of Anarchy} (PoA)~\citep{Koutsoupias:2009:WE:2296005.2296046} analysis is the predominant, powerful tool for quantifying the potential suboptimality of the outcomes/equilibria; it measures the impact the lack of coordination (strategic behaviour) has on the solution quality, by comparing it to the optimal and non-strategic solution.

Scheduling is one of the most influential problems in Algorithmic Game Theory
and has been studied extensively. In its most general form, the goal is to
schedule $m$ tasks to $n$ parallel machines with arbitrary processing times, in
order to minimize the makespan. In the front where payments are allowed, truthfulness comes at no extra cost given the strategic nature of the machines. \citet{Nisan:2001aa} first considered the mechanism design approach of the problem.
They prove that the well known VCG mechanism
achieves an $n$-approximation of the optimal makespan, while no truthful deterministic mechanism can achieve
approximation ratio better than $2$. The currently known best lower bound is $2.61$ \citep{KV13} while \citet{ADL12} proved the
tightness of the upper bound for anonymous mechanisms. With respect to randomized (truthful in expectation) mechanisms as well as fractional ones, the best known bounds are $(n+1)/2$ and $2-1/n$ \citep{Mu'alem:2007:SLB:1283383.1283506,Christodoulou:2010:MDF:1721837.1721854}. We note that the aforementioned lower bounds disregard computational feasibility and simply rely on the requirement for truthfulness.

In an attempt to get positive results when payments are not allowed in the scheduling context, \citet{K14} first considered the plausible assumption that the machines are bound by their declarations. This was influenced by the notion of impositions, according to which a mechanism can restrict the set of reactions available to a player after the outcome is chosen. This notion appeared in \citep{Nissim:2012:AOM:2090236.2090254,Fotakis2010} and was applied in facility location as well as digital goods pricing settings. The notion of winner imposition fits within the framework of approximate mechanism design without payments. A more powerful framework that is also very much related to this assumption is the notion of verification that appears in \citep{Nisan:2001aa,Auletta2004,Penna2014491}. The mechanisms in this context are allowed to use payments and simply give or deny payments to machines after they discover their true execution costs.
In particular, the mechanism receives limited information about the players' types after observing the computed solution.
Relevant works include \citep{Angel2009, Christodoulou2007} where the scheduling problem of selfish tasks is considered again under the assumption that the players who control the tasks are bound by their declarations.

\begin{paragraph}{Our Results} In this work we adopt the model of \citep{K14}.
For the case of scheduling a single task \citet{K14} proved that the
approximation ratio of any (randomized) truthful
mechanism is at least $(n+1)/2$ and gave a mechanism matching this bound, where
$n$ is the number of machines. When applied to many tasks, this mechanism
immediately implies an $n(n+1)/2$ approximation ratio for the makespan objective.
In \cref{sec:PoA} we provide a (non-truthful) algorithm which performs
considerably better than the best truthful mechanism; even the worst pure
equilibrium/outcome of our algorithm achieves an optimal makespan, i.e.\ our
algorithm has a pure PoA of $1$. Pure Nash equilibria constitute a very strong solution concept, that in general need not even exist, and they are also very desirable from a mechanism designer's perspective. For our algorithm we actually prove both existence of pure Nash equilibria and almost optimal performance guarantees for \emph{all} of them.\footnote{Clearly, pure equilibria are a refinement of mixed ones so, from the perspective of mixed Nash equilibrium analysis, our positive results can additionally be seen as an (almost) optimal Price of Stability guarantee (see also the discussion at the start of \cref{sec:PoS}).} If we run this algorithm independently for each
job, we get a task-independent and anonymous
algorithm yielding a PoA of $n$ for any number of tasks, which we show to be
asymptotically optimal. Next, revisiting truthfulness, in \cref
{sec:truthfulness} we also show that the mechanism inspired by the LP relaxation of the problem is provably truthful and achieves an optimal approximation ratio $1$ for the fractional scheduling problem of divisible tasks while providing an $n$-approximation when interpreted as a randomized mechanism. This almost matches the lower bound of $(n+1)/2$ for truthful mechanisms known from \citep{K14}. Finally, in \cref{sec:PoS} we briefly study the more optimistic objective of minimizing the makespan at the best possible equilibrium (instead of the worst one used in the Price of Anarchy metric) and show that the natural greedy algorithm achieves an optimal Price of Stability. \end{paragraph}

\section{Model and Notation}\label{sec:model}
We have a set $N=\sset{1,2,\dots,n}$ of unrelated parallel machines and $m$
tasks/jobs that need to be scheduled to these machines. Throughout the text we
assume that matrix $\vecc t\in\R^{n\times m}_{\geq 0}$ denotes the true
execution times, i.e.\ $t_{i,j}$ is the time machine $i$ needs to execute task $j$. This is private knowledge of each machine $i$. Let $\bids$ denote the corresponding (not necessarily true) \emph{declarations} of the machines for these costs.

A (randomized) \emph{allocation protocol} takes as input the machines' declarations $\bids$ and outputs an allocation $\vecc A$ of tasks to machines where $A_{ij}$ is a $0$--$1$ random variable indicating whether or not machine $i$ gets allocated task $j$ and $\vecc a$ is the corresponding probability distribution of allocation, i.e.\ $a_{i,j}=\probability{A_{i,j}=1}$ where of course $\sum_{i=1}^na_{i,j}=1$ for any task $j$.

We assume that machines are \emph{bound by their declarations}:
if a machine $i$ is allocated some task $j$, then she will
execute the task for time $\max\{t_{i,j},\bid{i,j}\}$. So, the expected cost/workload of machine $i$ is defined as
\begin{equation}
\label{eq:costplayer}
C_i(\bids|\vecc t_i)=\sum_{j=1}^m a_{i,j}(\hat{\vecc t}) \max\sset{\hat t_{i,j},t_{i,j}},
\end{equation}
while the \emph{makespan} is computed as the expected maximum execution time
$$
\mathcal M(\bids|\vecc t_i)=\expect_{\vecc{A}\sim \vecc a(\bids)}\left[\max_{i=1,\dots, n}\sum_{j=1}^m A_{i,j}\max\sset{\hat t_{i,j},t_{i,j}}\right].
$$
To simplify notation, whenever the true execution times $\vecc t$ are clear from the context we will drop them and simply use $C_i(\bids)$ and $\mathcal{M}(\bids)$.

The allocation protocol is called \emph{truthful}, or a truthful mechanism, if it does not give incentives to the machines to misreport their true execution costs.
Formally, for every machine $i$ and declarations matrix $\bids$,
$$C_i(\vecc t_i,\bids_{-i})\leq C_i(\bids),$$
where $(\vecc x_i,\vecc y_{-i})$ denotes the matrix of declarations where machine $i$ has deviated to  $\vecc x_i$ while all other machines report costs as in $\vecc y$. The \emph{approximation ratio} measures the performance of truthful mechanisms and is defined as the maximum ratio, over all instances, of the objective value (makespan) under that mechanism over the optimal objective value achievable by a centralized solution which ignores the truthfulness constraint.

If an allocation protocol is not truthful (we simply refer to it as an algorithm),
we measure its performance by the quality of its Nash equilibria: the states
from which no player has the incentive to unilaterally deviate. The \emph{Price
of Anarchy} (PoA)~\citep{Koutsoupias:2009:WE:2296005.2296046} is established as a meaningful benchmark and captures the
maximum ratio, over all instances, of the objective value of the worst equilibrium over that of the optimal centralized solution that ignores the machines' incentives. For most part of this paper we restrict attention to pure Nash equilibria where the machines make deterministic reports about their execution costs, and we will from now on refer to them simply as \emph{equilibria}. Then, the corresponding benchmark is called pure PoA. A more optimistic benchmark is the \emph{Price of Stability} (PoS)~\citep{Schulz2003,ADKTWR04} which compares the objective value of the \emph{best} equilibrium to the value of the optimal centralized solution.

The makespan objective is inherently different if we consider \emph{divisible tasks}, i.e.\ fractional allocations. In that case, each machine is allocated a portion of each task by the protocol and the makespan is computed as the maximum of the execution times of the machines, namely
$$
\mathcal M^f(\bids)=\max_{i=1,\dots , n}\sum_{j=1}^m \alpha_{i,j}(\bids)\max\sset{\hat t_{i,j},t_{i,j}}
$$
where $\alpha_{i,j}(\bids)\in [0,1]$ is the \emph{fraction} of task $j$ allocated to machine $i$. Again, it must be that $\sum_{i=1}^n\alpha_{i,j}=1$ for any task $j$.
Notice here that each fractional algorithm with allocation fractions $\vecc
\alpha$ naturally gives rise to a corresponding randomized integral algorithm
with allocation probabilities $\vecc a=\alpha$, whose makespan is within a
factor of $n$ from the fractional one\footnote{This is due to the fact that for any random variables $Y_1,Y_2,\dots,Y_n$ it is $\expect[\max_i Y_i]\leq \expect[\sum_i Y_i]=\sum_{i}\expect[Y_i]\leq n\max_{i}\expect[Y_i]$, and also $\max_{i}\expect[Y_i]\leq \expect[\max_i Y_i]$ due to the convexity of the $\max$ function.}, i.e.\ for any cost vector $\vecc t$
\begin{equation}
\label{eq:fractintegralineq}
\mathcal M^f(\vecc t) \leq \mathcal M(\vecc t) \leq n\cdot \mathcal M^f(\vecc t).
\end{equation}
Except when clearly stated otherwise, in this paper we deal with the integral version of the scheduling problem.

\paragraph{Social Welfare} An alternative objective, very common in the Mechanism Design literature, is that of optimizing \emph{social welfare}, i.e. minimizing the combined costs of all players: $\mathcal W(\bids)=\sum_{i=1}^n C_i(\bids)$.
It is not difficult to see\footnotemark[\value{footnote}] that the makespan and social welfare objectives are within a factor of $n$ away, whatever the allocation algorithm $\vecc a$ and the input costs $\bids$ might be:
\begin{equation}
\label{eq:makespanwelfareineq}
\mathcal M(\bids)\leq \mathcal W(\bids)\leq n\cdot \mathcal M(\bids).
\end{equation}
Also notice that for the special case of a single task, since the job is eventually allocated entirely to some machine, the two objectives coincide no matter the number of machines $n$, i.e.\ $\mathcal M(\bids)= \mathcal W(\bids)$. Because of that and the linearity of the social welfare with respect to the players' costs, it is easy to verify that all algorithms we present in this paper achieve optimal ratios of $1$ for that objective, both with respect to equilibrium/PoA and truthfulness analysis (e.g.\ \cref{th:PoAnmachines,th:lp_integral_n}). We will not mention that explicitly again in the remaining of the paper and rather focus on makespan minimization, which is a more challenging objective for our scheduling problem.

\section{Price of Anarchy}
\label{sec:PoA}
For clarity of exposition, we first describe our scheduling algorithm in the special case of just $n=2$ machines (and one task) before presenting the algorithm for the general case of $n\geq 1$. Since we treat the case of only one task in this section, we use $\bid{i}$ and $t_i$ to denote the declared and true execution time of machine $i$, respectively, and use $a_i$ to denote $i$'s allocation probability.

\subsection{Warm Up: The Case of Two Machines}
\label{sec:PoA2machines}
To simplify notation, throughout this section we will assume without loss of generality that $\bid{1}\leq \bid{2}$, i.e.\ the input to our algorithm is sorted in nondecreasing order.
Notice that the true execution times $\vecc t=(t_1,t_2)$ do not have to preserve this ordering, since the highest biding machine might very well in reality not have the fastest execution capabilities.

Our algorithm for the case of two machines, parameterized by two constants  $\lrg>2$, $c> 1$, and denoted by $\mathcal{A}_{\lrg,c}^{(2)}$, is defined by the allocation probabilities in \cref{fig:algo2machines}.
\begin{figure}[t]
\begin{equation*}
\begin{array}{l|ccc}
& a_1 & a_2 \\
\hline&&\\
\text{if}\quad \bid{1}=\bid{2} & \frac{1}{2} &  \frac{1}{2} \\
&&\\
\text{if}\quad \bid{1}<\bid{2}< c\cdot \bid{1} &\frac{1}{\lrg} &1-\frac{1}{\lrg}\\
&&\\
\text{if} \quad c\cdot \bid{1}\leq \bid{2} &1-\frac{1}{\lrg}\frac{\bid{1}}{\bid{2}} & \frac{1}{\lrg}\frac{\bid{1}}{\bid{2}}
\end{array}
\end{equation*}
\caption{Algorithm $\mathcal A^{(2)}_{\lrg,c}$ for scheduling a single task to two machines, parametrized by $\lrg>2$ and $c>1$. The probability that machine $i=1,2$ gets the task is denoted by $a_i$, and $\bid{1},\bid{2}$ are the reported execution times by the machines.}
\label{fig:algo2machines}
\end{figure}
Whenever parameter $c$ is insignificant in a particular context\footnote{In such a case, as it is for example in the statement of \cref{th:PoA_1task}, one can simply pick e.g.\ $c=1+\frac{1}{\lrg}$.}, we will just use $\mathcal{A}_{\lrg}^{(2)}$.

The main result of this section is the following theorem, showing that by choosing parameter $\lrg$ arbitrarily high, the above algorithm can achieve an optimal Price of Anarchy:
\begin{theorem}
\label{th:PoA_1task}
For the case of one task and two machines, algorithm $\mathcal{A}_{\lrg}^{(2)}$
has a (pure) Price of Anarchy of at most $1+\frac{1}{\lrg}$ (for any $\lrg>2$).
\end{theorem}

We break down the proof of \cref{th:PoA_1task} in distinct claims.
\begin{claim}
\label{claim:1dev2tasks}
At \emph{any} equilibrium $\bids$ the ratio of the two bids must be at least $c$, i.e.\ $\bid{2}\geq c\cdot \bid{1}$.
\end{claim}
\begin{proof}
Without loss assume $\bid{1}\neq 0$, since otherwise the claim is trivially true. First, assume for a contradiction that $\bid{1}<\bid{2}<c\cdot \bid{1}$.
Then the machine with larger report would have an incentive to deviate to bid $t_2'=\max\{c\bid{1}, t_2\}$:
\begin{equation*}
C_2(\bids)= \left(1-\frac{1}{\lrg}\right) \max\{\bid{2},t_2\}
> \frac{1}{\lrg} \bid{1}
=\frac{\bid{1}}{\lrg t_2'} \max\{t_2',t_2\}
=C_2(\bid{1},t_2')
\end{equation*}
where the inequality holds since $\lrg>2$ and the final two equalities hold because the deviating bid $t_2'$ equals $\max\{c\bid{1}, t_2\}$. Thus $\bids=(\bid{1},\bid{2})$ could not have been an equilibrium under the assumption that $\bid{1}<\bid{2}<c\cdot \bid{1}$.

A similar contradiction can be obtained for the remaining case of $\bid{1}=\bid{2}$. In this case, both machines have an incentive to deviate to a bid $t_1'$ such that $\frac{\bid{1}}{c}<t_1'< \bid{1}$, since
\begin{equation*}
C_1(\bids)= \frac{1}{2} \max\{\bid{1},t_1\}
\geq  \frac{1}{2} \max\{t_1',t_1\}
> \frac{1}{\lrg} \max\{t_1',t_1\}
= C_1(t_1',\bid{2}).
\end{equation*}
We can conclude that indeed $\bid{2}\geq c\bid{1}$ at any equilibrium.
\end{proof}

\begin{claim}
\label{claim:2dev2tasks}
At any equilibrium $\bids$ the machine with the larger report will never have underbid, i.e.\ $\bid{2}\geq t_2$.
\end{claim}
\begin{proof}
Assume for a contradiction that $\bid{2}< t_2$. Then
\begin{equation*}
C_2(\bids) =  \frac{\bid{1}}{\lrg \bid{2}} \max\{\bid{2},t_2\}
=  \frac{\bid{1}}{\lrg \bid{2}} t_2
>\frac{\bid{1}}{\lrg t_2} t_2
= C_2(\bid{1},t_2),
\end{equation*}
the first equality holding due to \cref{claim:1dev2tasks} and the last one because $t_2>\bid{2}\geq c \bid{1}$.
\end{proof}

\begin{claim}
\label{claim:3dev2tasks}
At any equilibrium $\bids$ the smaller bid is given by $\bid{1}= \min\{t_1,\frac{\bid{2}}{c}\}$.
\end{claim}
\begin{proof}
Assume for a contradiction that $\bid{1}\neq t_1'=\min\{t_1,\frac{\bid{2}}{c}\}$. Then, we will show that the lower bidding machine would have an incentive to deviate from $\bid{1}$ to $t_1'$.

Indeed, first consider the case when $\bid{1}<t_1'$. Then
\begin{equation*}
C_1(\bids) = \left(1-\frac{\bid{1}}{\lrg \bid{2}}\right) \max\{\bid{1},t_1\}
> \left(1-\frac{t_1'}{\lrg \bid{2}}\right) \max\{t_1',t_1\}\\
= C_1(t_1',\bid{2}).
\end{equation*}
In the remaining case of $\bid{1}> t_1'=\min\{t_1,\frac{\bid{2}}{c}\}$, because of \cref{claim:1dev2tasks} it must be that $t_1'=t_1<\bid{1}\leq \frac{\bid{2}}{c}$, thus
\begin{equation*}
C_1(\bids)= \left(1-\frac{\bid{1}}{\lrg \bid{2}}\right) \max\{\bid{1},t_1\}\\
= \left(1-\frac{\bid{1}}{\lrg \bid{2}}\right) \bid{1}\\
>  \left(1-\frac{t_1}{\lrg \bid{2}}\right) t_1\\
= C_1(t_1,\bid{2}).
\end{equation*}
where the inequality holds since $x\mapsto\left(1-\frac{x}{y}\right) x$ is a strictly increasing function for $x\in[0,\frac{y}{2}]$, and indeed $t_1<\bid{1}<\bid{2}<\frac{\lrg \bid{2}}{2}$.
\end{proof}

\begin{claim}
\label{claim:4dev2tasks}
At any equilibrium $\bids$ bidding must preserve the relative order of the true execution times, i.e. $t_1\leq t_2$.
\end{claim}
\begin{proof}
For a contradiction assume that $t_2< t_1$, and first consider the case when $t_2< \bid{1}$.
If we pick $t_2'\in\left(\frac{\bid{1}}{c},\bid{1} \right)$ we have
\begin{equation*}
C_2(\bids)= \frac{\bid{1}}{\lrg \bid{2}} \max\{\bid{2},t_2\}
= \frac{\bid{1}}{\lrg}
>  \frac{1}{\lrg}\max\{t_2', t_2\}
= C_2(\bid{1},t_2'),
\end{equation*}
meaning that the higher bidding machine would have an incentive to deviate from $\bid{2}$ to $t_2'$.

For the remaining case of $\bid{1}\leq t_2< t_1$, first note that if $t_1\leq\frac{\bid{2}}{c}$ then by \cref{claim:3dev2tasks} we would immediately derive that $\bid{1}=t_1$, which is a contradiction. Hence, by \cref{claim:1dev2tasks} we can assume that $\bid{1}=\frac{\bid{2}}{c}< t_1$. Then, if $\bid{2}>t_1$  we have that
\begin{equation*}
C_1(\bids)= \left(1-\frac{\bid{1}}{\lrg \bid{2}}\right) \max\{\bid{1},t_1\}
= \left(1-\frac{\bid{1}}{\lrg\bid{2}}\right) t_1
>  \frac{1}{\lrg}t_1
= C_1(t_1,\bid{2}),
\end{equation*}
the inequality holding because $\frac{\bid{1}}{\bid{2}}\frac{1}{\lrg}\leq\frac{1}{\lrg}<\frac{1}{2}$, and if $\bid{2}\leq t_1$ then, in the same way, for $t_1'=\max\{t_1,c\bid{2}\}$
\begin{equation*}
C_1(\bids)
>  \frac{1}{\lrg}t_1
\geq \frac{\bid{2}}{\lrg}
=\frac{\bid{2}}{\lrg t_1'}\max\{t_1',t_1\}
= C_1(t_1',\bid{2}).
\end{equation*}
\end{proof}

\begin{proof}[Proof of \cref{th:PoA_1task}]
\Cref{claim:1dev2tasks,claim:2dev2tasks,claim:3dev2tasks,claim:4dev2tasks} imply that the makespan (and thus also the social cost since we have a single task) of any allocation at equilibrium can be bounded by
\begin{align*}
\mathcal{M}(\bids) &= \left(1-\frac{\bid{1}}{\lrg \bid{2}}\right) \max\{\bid{1},t_1\}+\frac{\bid{1}}{\lrg \bid{2}}\max\{\bid{2},t_2\}\\
&\leq \max\{\bid{1},t_1\}+\frac{1}{\lrg}\bid{1} \\
&\leq \left(1+\frac{1}{\lrg}\right)t_1,
\end{align*}
where $t_1$ is the optimal makespan.

Also, it is important to mention that it can be verified that there exists at least one (pure Nash) equilibrium, e.g.\ reporting $\bid{1}=t_1$ and  $\bid{2}=\max\{\lrg c\cdot t_1, t_2\}$.
\end{proof}

\subsection{The General Case}
The algorithm for two machines (and a single task) can be naturally generalized to the case of any number of machines $n\geq 2$.
We note that the essence of the techniques and the core ideas we presented in \cref{sec:PoA2machines} carry over to the general case. So, for clarity of exposition, we only give the definition of the algorithm here and the proof can be found in \cref{append:PoAnmachines}.

To present our algorithm $\mathcal{A}_{\lrg,c}$ we first need to add some notation.
We use $\bmin$ and $\bsec$ to denote the smallest and second smallest declarations in $\bids$, and $N_{\min}$, $N_{\second}$ the corresponding sets of machine indices that make these declarations. (If $N=N_{\min}$, i.e.\ all machines make the same declaration we define $\bsec=\bmin$).
Also let $n_{\min}=|N_{\min}|$ and $n_{\second}=|N_{\second}|$.

Our main algorithm $\mathcal{A}_{\lrg,c}$ for the case of one task and $n$ machines, parameterized by $\lrg>2(n-1)$, $c>1$, is defined by the allocation probabilities $a_i$ for each machine $i\in N$ given in \cref{fig:algorithm_n_machines}.

\begin{figure}[t]
\vspace{-0.3cm}
\begin{equation*}
\begin{array}{l | c c c}
 & \negthickspace\negthickspace i\in N_{\min} & \negthickspace\negthickspace\negthickspace i\in N_{\second} & \negthickspace\negthickspace\negthickspace i\in N\negthickspace\setminus\negthickspace (N_{\min}\negthickspace\cup\negthickspace N_{\second})\\
\hline&&\\
\text{if}\;\; \bmin=\bsec &\frac{1}{n} &  \frac{1}{n} & \frac{1}{n}\\
&&\\
\text{if}\;\; \bmin<\bsec< c\cdot \bmin &\left(\frac{1}{\lrg} \right)/n_{\min}&
\left(1-\frac{1}{\lrg}\right)/n_{\second}&0\\
&&\\
 \text{if}\;\;  \bsec\geq c\cdot \bmin &\left(1-\displaystyle\sum_{k\in N\setminus N_{\min}}\frac{\bmin}{\lrg\cdot \bid{k}}\right)/n_{\min} & \frac{\bmin}{\lrg\cdot\bidi} & \frac{\bmin}{\lrg\cdot \bidi}
\end{array}
\end{equation*}
\caption{Algorithm $\mathcal A_{\lrg,c}$ for scheduling a single task to $n\geq 2$ machines, parametrized by $\lrg>2(n-1)$ and $c>1$. The first and second highest reported execution times by the machines are denoted by $\bmin$ and $\bsec$ respectively, while $N_{\min}$, $N_{\second}$ denote the corresponding sets of machine indices, and $n_{\min}$, $n_{\second}$ their cardinalities.}
\label{fig:algorithm_n_machines}
\end{figure}

As the following theorem suggests, by picking a high enough value for $\lrg$ the above algorithm can achieve an optimal performance under equilibrium:
\begin{theorem}
\label{th:PoAnmachines}
For the problem of scheduling one task without payments to $n\geq 2$ machines,
algorithm $\mathcal{A}_{\lrg}$ has a (pure) Price of Anarchy of at most $1+\frac
{n-1}{\lrg}$ (for any $\lrg>2(n-1)$).
\end{theorem}
\paragraph{Multiple Tasks}
It is not difficult to extend our single-task algorithm and the result of \cref{th:PoAnmachines} to get a task-independent, anonymous algorithm with a pure PoA of $n$ for any number of tasks $m\geq 1$: simply run $\mathcal A_L$ \emph{independently} for each job. Then, the equilibria of the extended setting correspond exactly to players not having an incentive to deviate for any task/round, and the approximation ratio of $1+\frac{n-1}{\lrg}$ with respect to the minimum cost $\min_i t_{i,j}$ at every such round $j=1,\dots,m$, guarantees optimality with respect to the social welfare and thus provides indeed a worst-case $n$-approximation for the makespan objective (see \cref{eq:makespanwelfareineq}).
The following result shows that this is in fact (asymptotically) optimal for
the natural class of anonymous\footnote{A scheduling algorithm is called \emph
{anonymous} if it does not discriminate among the players (with respect to
their identities). Formally,
any permutation of the players
results in the same permutation to the output/allocation
of the algorithm, i.e.,
for any permutation $\pi$ of the set of player
indices $\ssets{1,2,\dots,n}$ and any input $\vecc t$, it is
$a(\pi\cdot \vecc t)=\pi\cdot a(\vecc t)$, where the dot product $\pi \cdot
\vecc t$ denotes the permutation of the rows of matrix $\vecc t$ according to
$\pi$.}, task-independent algorithms:

\begin{theorem} \label{th:lower_bound_anonymous_taskind} No task-independent,
anonymous algorithm for scheduling without payments on $n$ machines can have a
(pure) Price of Anarchy better than $\frac{n}{2}-o(1)$. 
\end{theorem}
\begin{proof} 
Fix an algorithm that allocates tasks by running an anonymous
single-task algorithm $\mathcal A_j$, independently for each task $j$. Each algorithm
$\mathcal A_j$ takes as input a
declared cost vector $\hat{\vecc t}= (\hat t_1, \dots, \hat t_n)^\top$  by the
machines, where $\hat t_i$ is the report of machine $i$ for the task, $i=1,\dots,n$. Of
course, there is also an underlying \emph{true} cost vector $\vecc t$. We will
study cost vectors $\vecc t^{j}$ under which the corresponding task $j$ has a high execution time $M>n^3$ for all
machines except two of them, namely $i=1$ and $i=j$. Formally, for any
$j=1,\dots,n$ define:
$$  t^{j}_i= \begin{cases} 1, &\text{if}\;\;
i=1\;\;\text{or}\;\;i=j,\\ M, &\text{otherwise}. \end{cases} $$
Let $\hat{\vecc
t}^{j}$ be a (pure) Nash equilibrium strategy vector under algorithm
$\mathcal A_j$, when
the true execution costs are $\vecc t^j$.

First
we notice that we can safely assume that, under any such true instance $\vecc
t^j$, the probability of $\mathcal A_j$
allocating the
task to one of the fast
machines, at equilibrium, has to be at least
$a_1(\hat{\vecc t}^{j})+a_j(\hat{\vecc t}^{j})\geq 1-\frac{1}{n^2}$.
Otherwise the expected makespan would be at least $
\left(1-\frac{1}{n^2}\right)\cdot 1 + \frac{1}{n^2}\cdot M>n+1-\frac{1}{n^2}>\frac
{n}{2}$ (while the optimal one is $1$, by assigning the task to machine $i=1$)
thus easily proving the Price of Anarchy lower bound of the theorem.
So, at equilibrium $\hat{\vecc t}^{j}$, one of machines $i=1$ and $i=k$ has to
receive the job with probability at least $\frac{1}{2}\left(1-\frac{1}
{n^2}\right)=\frac{1}
{2}-\frac{1}{2n^2}$. Assume that the machine with this property is
$i=1$; otherwise, due
to the anonymity of algorithm $\mathcal A_j$, it is not difficult to see that the
vector that results from $\hat{\vecc t}^{j}$ by swapping its entries between
positions $i=1$ and $i=j$, and keeping all other machines fixed, is also a Nash
equilibrium and furthermore has the desired property.

Now we are ready to construct the bad instance for our task-independent,
anonymous algorithm. Consider $m=n$ jobs, job $j=1,2,\dots,n$ having true execution time
vector $\vecc t^{j}$. By the previous analysis for the single-task
algorithms $\mathcal A_j$, $\hat{\vecc t}=(\hat{\vecc t}^{1},\hat{\vecc t}^
{2},\dots,\hat
{\vecc t}^{n})$ is a Nash equilibrium for our multi-task algorithm, and the
probability
that machine $i=1$ is the one allocated task $j$ is at least
$\frac{1}{2}-\frac{1}{2n^2}$, for every step $j$. So, by
utilizing a standard Chernoff bound\footnote{
Use the tail inequality
$\probability{X\leq (1-\delta)\mu}\leq e^{-\frac{\delta^2} {2}\mu}$
where $X$ is the number of tasks that are allocated to machine $i=1$,
$\mu=n\left(\frac{1}{2}-\frac{1}{2n^2}\right)$, and choose $\delta\leq n^
{-\frac{1}{3}}$.
}
it is straightforward to see that the probability that machine $i=1$ gets more
than $(1-o(1))\cdot n\left(\frac{1}{2}-\frac{1}{2n^2}\right)=(1-o(1))\left(\frac{n}{2}-o(1)\right)$ jobs in total is at least $1-o(1)$,
resulting to an expected makespan of at least $\frac{n}{2}-o(1)$. The expected
optimal makespan is $1$: just allocate task $j$ to machine $j$,
for all $j=1,2,\dots,n$.
\end{proof}

Although anonymity is definitely a very natural assumption in the scheduling
domain (see e.g.\ the discussion in \citep{ADL12}), the above results do not
rule out the possibility of an improved Price of Anarchy if the requirements of
\cref{th:lower_bound_anonymous_taskind} are relaxed, and as a matter of fact we
find the design of such algorithms, that will take into consideration
non-trivial correlations between tasks and among machines, to be one of the most
interesting and challenging directions for future work. However keep in mind
that, as \cref{th:lower_bound_taskind} in \cref{append:lower_bound_taskind}
demonstrates, just relaxing the anonymity constraint would not
suffice in order to hope for constant-factor performance.

\section{Truthful Mechanisms}
\label{sec:truthfulness}
In this section we turn our attention to truthful algorithms for many tasks and provide a mechanism that achieves approximation ratio $n$, almost matching the $(n+1)/2$ known lower bound on truthfulness \citep{K14}. The best known ratio before our work was $n(n+1)/2$, achieved by running the algorithm of \citet{K14} independently for each task. Unfortunately this guarantee turns out to be tight for the particular algorithm
(see \cref{sec:tight_lower_K14} for a bad instance),
thus here we have to devise more involved, \emph{non} task-independent mechanisms.

\subsection{The LP mechanism}
It is a known fact that the LP relaxation of a problem can be a useful tool for designing mechanisms (both randomized and fractional). What the LP mechanism essentially does is that it forces all machines to have equal expected (minimum) cost. This aligns the goal of the designer with the objective of the machines, thus making it easier for truthfulness to be guaranteed; similar in spirit is the {\sc{Equal Cost}} mechanism of \cite{FT13}. We recall that the LP relaxation for the scheduling problem is as shown in \cref{fig:LP}.

\begin{figure}[ht]
\vspace{-0.5cm}
\begin{align*}
& & \text{minimize} \;\; \mu & & \\
&\forall j: &\sum_{i=1}^n \alpha_{i,j}&=1&&\textrm{ (each task is allocated entirely)}\\
&\forall i: &\mu-\sum_{j=1}^m \alpha_{i,j}t_{i,j}&\geq 0&&\textrm{ (the cost of each machine does not exceed makespan)}\\
&\forall i,j: &\alpha_{i,j}&\geq0&&\textrm{ (the allocation probabilities are non-negative)}
\end{align*}
\caption{The LP relaxation for the scheduling problem. Our LP mechanism is defined by an optimal solution $\alpha_{i,j}^{\text{LP}}(\vecc t)$ to this program.}
\label{fig:LP}
\end{figure}

We denote an \emph{optimal} solution\footnote{Notice that although $\mu_{\vecc t}^{\text{LP}}$ is unique, there might be various allocation fractions $\alpha_{i,j}$ that give rise to the optimal makespan $\mu_{\vecc t}^{\text{LP}}$, in which case we can choose an arbitrary one for $\alpha^{\text{LP}}_{\vecc t}$, e.g.\ take the lexicographically smaller.} to the above LP by $\alpha^{\text{LP}}(\vecc t)$, and the optimal objective value by $\mu^{\text{LP}}_{\vecc t}$ (dropping the $\text{LP}$ superscript whenever this is clear from the context). The vector $\alpha^{\text{LP}}(\vecc t)$ can be straightforwardly interpreted as allocation probabilities or allocation fractions giving rise to a randomized or a fractional mechanism, respectively. We refer to the corresponding mechanisms as the LP randomized and the LP fractional mechanism. In \cref{thm:lp_truthful} we show that both mechanisms are truthful, hence, we can think of $\mu_{\vecc t}^{\text{LP}}$ as corresponding to the maximum (expected) cost/workload perceived by any machine.

It is a simple observation that in an optimal solution the workload must be fully balanced among all machines and that $\mu^{\text{LP}}$ can only increase when all execution times increase, i.e.\ $\mu_{\vecc t}^{\text{LP}}\leq \mu_{\vecc t'}^{\text{LP}}$ for $\vecc t \leq \vecc t'$ (pointwise).

Note that the proof of \cref{thm:lp_truthful} is identical in both cases where the $\alpha$ corresponds to fractions or allocation probabilities. Hence, the result holds for both the LP randomized and the LP fractional mechanisms. Note that a strategy is (weakly) dominant for a machine if the strategy will result in a cost at most equal to the cost of any other strategy, for every choice of strategies/declarations of the other machines.

\begin{theorem}\label{thm:lp_truthful}
Under the LP (fractional or randomized) mechanism, truthfully reporting the execution times is a (weakly) dominant strategy for every machine.
\end{theorem}
\begin{proof}
Recall that $\vecc t$ and $\bids$ denote the true and (some) declared execution
times, respectively, of the machines.
Fix some machine $i$ and define matrix $\vecc t^{\max_i}$ as follows: row $i$, $\vecc t^{\max_i}_i$, is  the vector of point-wise maxima between true and declared times for machine $i$, that is
$$
\vecc t^{\max_i}_{i}=(\max\{\hat t_{i,1},t_{i,1}\}, \max\{\hat t_{i,2},t_{i,2}\}, \ldots, \max\{\hat t_{i,m},t_{i,m}\}),
$$
while every other row $k\neq i$ is $\vecc t^{\max_i}_k=\vecc{\hat{t}}_k$, i.e.\ $\vecc t^{\max_i}=(\vecc t^{\max_i}_i,\bids_{-i})$. Seen as a matrix of declarations,  $\vecc t^{\max_i}$ corresponds to machine $i$'s deviation from $\bids_i$ to $\vecc t^{\max_i}_i$, while each other machine makes declarations according to $\bids$.
Then we can derive the following:
$$
\sum_{j=1}^{m} \alpha_{k,j}(\hat{\vecc{t}}) t^{\max_i}_{k,j}=\sum_{j=1}^{m} \alpha_{k,j}(\hat{\vecc{t}}) \hat t_{k,j}=\mu_{\hat{\vecc t}}=\sum_{j=1}^{m} \alpha_{i,j}(\hat{\vecc{t}}) \hat t_{i,j}\leq \sum_{j=1}^{m} \alpha_{i,j}(\hat{\vecc{t}}) t^{\max_i}_{i,j}
$$
and thus from the optimality of the LP solutions it must be that
$$\mu_{\vecc t^{\max_i}}\leq \max_{l=1,\dots,n}\sset{\sum_{j=1}^{m} \alpha_{l,j}(\hat{\vecc{t}}) t^{\max_i}_{l,j}}=\sum_{j=1}^{m} \alpha_{i,j}(\hat{\vecc{t}}) t^{\max_i}_{i,j}.$$
Bringing everything together and taking into consideration that $(\vecc t_i,\hat{\vecc t}_{-i})\leq \vecc{t}^{\max_i}$ we get
$$
\negthickspace\negthickspace\negthickspace\negthickspace
C_i(\hat{\vecc t})
=\sum_{j=1}^{m} \alpha_{i,j}(\hat{\vecc{t}}) \max\sset{\hat t_{i,j}, t_{i,j}}
\geq \mu_{\vecc t^{\max_i}}
\geq \mu_{(\vecc t_i,\hat{\vecc t}_{-i})}
=\sum_{j=1}^{m} \alpha_{i,j}(\vecc t_i,\hat{\vecc t}_{-i}) t_{i,j}
=C_i(\vecc t_i,\hat{\vecc t}_{-i}),
$$
which shows that indeed, whatever the declarations of the other machines $\hat{\vecc t}_{-i}$, machine $i$ is always (weakly) better off by truthfully reporting $\vecc t_i$.
\end{proof}
\Cref{thm:lp_truthful} gives rise to the following two results.
\begin{theorem}
The LP fractional mechanism has approximation ratio $1$ for the fractional scheduling problem without money, for any number of machines and tasks.
\end{theorem}
As discussed in \cref{sec:model}, by \eqref{eq:fractintegralineq} we know that the above performance guarantee can deteriorate at most by a factor of $n$ when we use the fractions as allocation probabilities for the integral case:
\begin{theorem}
\label{th:lp_integral_n}
The LP randomized mechanism has approximation ratio at most $n$ for (integrally) scheduling any number of tasks to $n$ machines without money.
\end{theorem}

\subsection{The Proportional Mechanism}
In this section we briefly consider the proportional mechanism which allocates to each machine $i$ a $t_i^{-1}/\sum_{k=1}^n t_k^{-1}$ fraction of the task or probability of getting the task, respectively, depending on whether we consider the randomized or the fractional variant. In \citep{K14} it was shown that this algorithm is truthful and that its approximation ratio for randomized allocations of a single task is $n$. With the following theorem we wish to stress the difference between fractional and (randomized) integral allocations. The theorem is about the fractional case and proves the optimality of the proportional mechanism for scheduling one task without payments.

\begin{theorem}\label{thm:proportional}
The proportional mechanism has an optimal approximation ratio of $1$ for the \emph{fractional} scheduling problem of a single task. For $m$ tasks the approximation ratio increases to at least $m$.
\end{theorem}
\begin{proof}
First consider the case of a single task.
Under the proportional mechanism each machine $i$ is allocated an $\alpha_i=t_i^{-1}/\sum_{k=1}^n t_k^{-1}$ fraction of the task and executes it for time $t_i$. Hence, all machines have the same execution time (makespan) equal to
\begin{equation*}
\mathcal{M}^f = \frac{1}{\sum_{i} t_i^{-1}}.
\end{equation*}
In the optimal fractional allocation, all machines will have the same execution
cost, otherwise we could remove an $\varepsilon$ fraction of the task from a machine with high makespan and allocate it to a machine with small workload, hence reducing the makespan. This implies that for the optimal allocation fractions $\{\alpha^*_i\}_{i=1}^n$ it holds that $t_1 \alpha_1^*=t_2 \alpha_2^*= \cdots =t_n \alpha_n^*$. Thus, for any machine $i$ it  is $\alpha^*_i=\frac{t_1}{t_i}\alpha^*_1$, and by the fact that $\sum_i \alpha^*_i =1$ we get that $\sum_i t_i^{-1}=\frac{1}{\alpha^*_1 t_1}$ and so
\begin{equation*}
\mathcal{M}^f=\alpha^*_1 t_1 =\text{OPT}^f.
\end{equation*}

However the above result does not generalize to the case of many tasks, where the proportional mechanism is run independently for each task. First, it is easy to see that this independence preserves truthfulness. Secondly, it also preserves the optimality of the proportional mechanism with respect to social welfare. But regarding makespan, we will prove a lower bound of $m$.
Consider the following instance with $m$ tasks and $m$ machines. Every machine $i$ can execute all tasks in time $1$, except from the $i$-th task that can be run very quickly in time $M^{-1}$, where $M\gg 1$. Formally,
$$
t_{ij}=
\begin{cases}
M^{-1},& j=i\\
1, & j\neq i.
\end{cases}
$$
Then, the proportional mechanism computes allocation fractions
$$
\alpha_{ij}=
\begin{cases}
\frac{M}{M+m-1},& j=i\\
\frac{1}{M+m-1}, & j\neq i,
\end{cases}
$$
which results to a makespan of
$$
\sum_{j=1}^m\alpha_{ij}t_{ij}=M^{-1}\frac{M}{M+m-1}+(m-1)\frac{1}{M+m-1}=\frac{m}{M+m-1}.
$$
On the other hand, the allocation that assigns each task $j$ to its fastest
machine $j$, i.e.\ $\alpha_{ij}=1$ for $i=j$, results to a makespan of $M^{-1}$,
giving an approximation ratio lower bound of
$
\frac{Mm}{M+m-1}\to m
$
as $M\to\infty$.
\end{proof}

\section{Price of Stability and Mixed Equilibria}
\label{sec:PoS}

In this section we attempt a more optimistic approach regarding the problem of scheduling without payments. The Price of Stability can be used to measure the best possible performance of a given mechanism in practice, i.e., it can provide some indication on what is its potential power or, in other words, what is the best we can hope for given a specific mechanism. There are settings where this can be very useful, e.g. if the context allows the mechanism to suggest a specific strategy profile to the players, and the players then decide if it is in their best interest to follow the suggested strategies.

We consider the benchmark of the \emph{best} (mixed Nash) equilibrium and prove that the following, most natural greedy algorithm can achieve optimality: allocate each task independently to the machine declaring the minimum cost (breaking ties arbitrarily).

\begin{theorem}\label{thm:greedy_pos}
The Price of Stability of the Greedy algorithm is $1$ for scheduling without money any number of tasks to any number of machines.
\end{theorem}
\begin{proof}
We will prove the stronger statement that every feasible (integral) allocation of $m$ tasks to $n$ machines can arise at some (mixed) Nash equilibrium of the Greedy mechanism, from which the theorem immediately follows. First we observe the following, not difficult to prove fact:
\begin{fact}
Fix some task and a nonnegative constant $T$. If all but one machine $i^*$ play the mixed strategy of declaring (independently) a value $x\in [T,\infty)$ with probability distribution $F(x)=1-\sqrt[n-1]{T/x}$, and the remaining machine $i^*$ declares deterministically any value $x^*>T$, then $i^*$'s declared cost will be the minimum among all declarations for the task with probability $(1-F(x^*))^{n-1}$. Thus, under the Greedy mechanism, machine $i^*$ will incur an expected cost of at least $x^*(1-F(x^*))^{n-1}=T$ for executing the particular task.
\end{fact}

Consider any true instance $\vecc t$ of the scheduling without payments problem and fix a particular allocation $A$ of tasks to machines. For each task $j$ let $i^*_j\in N$ denote the machine that $j$ is assigned to under allocation $A$. Consider now the following strategy profile: machine $i^*_{j}$ truthfully declares her cost for task $j$, i.e.\ $\bid{i^*_j,j}=t_{i^*_j,j}\equiv T_j$ deterministically, while all other machines $k\neq i^*_j$ each play (independently) a mixed strategy of declaring for task $j$ a higher value of $\bid{k,j}=x>T_j$ with cumulative distribution $F_j(x)=1-\sqrt[n-1]{T_j/x}$. We argue that this constitutes a Nash equilibrium for the Greedy mechanism. Indeed, under these declarations the Greedy algorithm allocates each task $j$ to machine $i^*_j$ for a cost (of executing this task) of $T_j$. But by the previous fact, machine $i^*_j$ cannot avoid incurring at least the same cost if misreporting any higher execution time $x^*>T_j$, while if she underbids $x^*<T_j$ she will obviously still get the task for effectively the same cost of $\max\ssets{x^*,t_{i^*_j,j}}=T_j$.
\end{proof}

\paragraph{Acknowledgements} We thank the anonymous reviewers for their careful reading of our manuscript and their valuable feedback.

\bibliographystyle{abbrvnat}
\bibliography{AnarchyWithoutPayments}

\appendix

\section{Proof of \texorpdfstring{\cref{th:PoAnmachines}}{Theorem 2}}
\label{append:PoAnmachines}
Analogously to the proof of \cref{th:PoA_1task}, we break down our exposition
into distinct claims.

\begin{claim}\label{claim:sec_bid}
At any equilibrium $\bids$ it holds that $\bsec\geq c\cdot \bmin$, i.e., the minimum declaration differs from the remaining declarations by at least a factor of $c$.
\end{claim}
\begin{proof}
Consider an equilibrium $\bids$ and assume for a contradiction that
$\bmin<\bsec< c\cdot\bmin$ at an equilibrium. Consider a machine $i\in N_{\second}$ and her deviation to $\bidi'=\max\{c\cdot \bidi,t_i\}$. It holds that
\begin{align*}
C_i(\bids)&= \left(1-\frac{1}{\lrg}\right) \frac{\max\{\bsec,t_i\}}{n_{\second}}
>  \frac{1}{\lrg} \bmin
=\frac{\bmin}{\lrg\cdot \bidi'} \bidi'
=\frac{\bmin}{\lrg \cdot\bidi'} \max\{\bidi',t_i\}
\geq C_i(\bidi',\bids_{-i})
\end{align*}
where the first inequality holds since $\bsec>\bmin$, and $\lrg>2(n-1) \geq n_
{\second}+1$ by assumption, and the second inequality holds because $i$'s cost will be either $\frac{\bmin}{\lrg\cdot\bidi'}\bidi'$ or $0$ depending on whether there exists an other declaration in $(\bmin,c\cdot\bmin)$. If $\bmin=\bsec$, i.e.\ all machines declare $\bmin$, consider any machine $i$ and let $\frac{\bmin}{c}<\bidi'<\bmin$. Then
\begin{equation*}
C_i(\bids)=\frac{1}{n}\max\{\bmin,t_i\}
>  \frac{1}{\lrg}\max\{\bmin,t_i\}
\geq  \frac{1}{\lrg}\max\{\bidi',t_i\}
= C_i(\bidi',\bids_{-i})
\end{equation*}
where the first inequality holds since $\lrg>2(n-1)\geq n$ for $n\geq 2$. We can conclude that, at any equilibrium of the game, $\bsec$ has to be at least as large as $c\cdot \bmin$. An immediate consequence is that all non-minimal declarations are at least equal to $c\cdot \bmin$.
\end{proof}

\begin{claim}\label{claim:bid>=true}
At any equilibrium $\bids$ it holds that $\bidi \geq t_i$ for any machine $i\in N\setminus N_{\min}$, i.e.\ any machine whose declaration is not the minimum one, declares a value at least equal to its true execution cost $t_i$.
\end{claim}
\begin{proof}
Consider an equilibrium $\bids$ and assume for a contradiction that $\bidi< t_i$ for some machine $i\notin N_{\min}$. Recall that from \cref{claim:sec_bid} we have that $\bidi\geq c\cdot \bmin$. But deviating from $\bidi$ to $t_i$ is always beneficial for this machine, as
\begin{equation*}
C_i(\bids)= \frac{\bmin}{\lrg\cdot\bidi} \max\{\bidi,t_i\}
=  \frac{\bmin}{\lrg\cdot\bidi} t_i
>\frac{\bmin}{\lrg\cdot t_i} t_i
= C_i(t_i,\bids_{-i}),
\end{equation*}
which proves our claim.
\end{proof}

\begin{claim}\label{claim:n_b}
At any equilibrium $\bids$ it holds that $n_{\min}=1$, i.e.\ only one machine makes the minimum declaration.
\end{claim}
\begin{proof}
Consider an equilibrium $\bids$ and assume for a contradiction that $n_{\min}>1$. Let $i$ be any machine such that $i\in N_{\min}$. It holds that
\begin{align*}
C_i(\bids)&= \left(1-\sum_{k\in N\setminus N_{\min}}\frac{\bmin}{\lrg\cdot \bidk}\right) \frac{\max\{\bmin,t_i\}}{n_{\min}}\\
&> \left(1-\sum_{k\in N\setminus N_{\min}}\frac{\bmin}{\lrg\cdot\bmin}\right) \frac{\max\{\bmin,t_i\}}{n_{\min}}\\
&= \left(1-\frac{n-n_{\min}}{\lrg}\right) \frac{\max\{\bmin,t_i\}}{n_{\min}}\\
&> \frac{1}{\lrg}\max\{\bmin,t_i\}\\
&\geq  \frac{\bmin}{\lrg}\\
&= \frac{\bmin}{\lrg\cdot\max\{\bsec,t_i\}}\max\{\bsec,t_i\}\\
&= C_i(\max\{\bsec,t_i\},\bids_{-i})
\end{align*}
where the first inequality holds since $\bidk>\bmin$ for all $k\in N\setminus N_{\min}$, the second inequality holds since $\lrg> 2(n-1)\geq n$ by assumption and the final equality holds since there would be at least one other declaration $\bmin$ and since by \cref{claim:sec_bid} we have that $\bsec\geq c\cdot \bmin$.
\end{proof}

\begin{claim}\label{claim:first_bid}
At any equilibrium $\bids$ it holds that $\bmin= \min\{t_i,\frac{\bsec}{c}\}$ for machine $i= N_{\min}$, i.e., if $i$ is the machine that makes the minimum declaration then that declaration is equal to $\min\{t_i,\frac{\bsec}{c}\}$.
\end{claim}
\begin{proof}
Consider an equilibrium $\bids$ and note that $i$ is well defined from \cref{claim:n_b}. Assume for a contradiction that $\bmin< \min\{t_i,\frac{\bsec}{c}\}$. We have
\begin{align*}
C_i(\bids)&= \left(1-\displaystyle\sum_{k\in N\setminus N_{\min}}\frac{\bmin}{\lrg\cdot\bidk}\right) \frac{\max\{\bmin,t_i\}}{n_{\min}}\\
&= \left(1-\displaystyle\sum_{k\in N\setminus N_{\min}}\frac{\bmin}{\lrg\cdot\bidk}\right) t_i\\
&> \left(1-\displaystyle\sum_{k\in N\setminus N_{\min}}\frac{\min\{t_i,\frac{\bsec}{c}\}}{\lrg\cdot\bidk}\right) t_i\\
&= C_i\left(\min\left\{t_i,\frac{\bsec}{c}\right\},\bids_{-i}\right),
\end{align*}
where we have used $n_{\min}=1$ proved in \cref{claim:n_b}. Now, it suffices to show that $\bids$ is not an equilibrium if $\bmin> \min\{t_i,\frac{\bsec}{c}\}$. We know from \cref{claim:sec_bid}, that at any equilibrium $\bmin\leq \frac{\bsec}{c}$, so we only need to consider the case where $t_i<\bmin\leq \frac{\bsec}{c}$. Then
\begin{align}
C_i(\bids|\vecc t)&= \left(1-\displaystyle\sum_{k\in N\setminus N_{\min}}\frac{\bmin}{\lrg\cdot\bidk}\right) \frac{\max\{\bmin,t_i\}}{n_{\min}}\notag\\
&= \left(1-\displaystyle\sum_{k\in N\setminus N_{\min}}\frac{\bmin}{\lrg\cdot\bidk}\right) \bmin\label{eq:aux2}\\
&>  \left(1-\displaystyle\sum_{j\in N\setminus N_{\min}}\frac{t_i}{\lrg\cdot\bidk}\right) t_i\notag\\
&= C_i(t_i,\bids_{-i}|\vecc t),\notag
\end{align}
where the inequality holds since $2\bmin\sum_{k\in N\setminus N_{\min}}\frac{1}{\lrg\cdot\bidk}< 2\bmin \frac{n-1}{\lrg\bmin}<\frac{2(n-1)}{\lrg}<1$ and the derivative of \eqref{eq:aux2} with respect to $\bmin<\min_k t_k$ is positive for $\lrg> 2(n-1)$.
\end{proof}

\begin{claim}\label{claim:t(b)}
At any equilibrium $\bids$ it holds that $t_i\leq t_k$ for machine $i= N_{\min}$ and any machine $k\in N$, i.e., the true execution cost of the machine that makes the minimum declaration is indeed at most equal to the true execution cost of any other machine.
\end{claim}
\begin{proof}
Consider an equilibrium $\bids$ and assume for a contradiction that there exists some machine $k\in N\setminus N_{\min}$ such that $t_k<t_i$. We consider cases depending on the relative order of $\bmin$, $t_i$, $\bidi$, and $t_k$. Assume that $t_k< \bmin$ and let $\frac{\bmin}{c}<\bidk'<\bmin$. Then
\begin{equation*}
C_{k}(\bids) = \frac{\bmin}{\lrg\cdot\bidk} \max\{\bidk,t_k\}
= \frac{\bmin}{\lrg}
>  \frac{1}{\lrg}\max\{\bidk', t_k\}
= C_{i}(\bidi',\bids_{-i})
\end{equation*}
where the second equality holds by \cref{claim:bid>=true}. This implies that at any equilibrium $\bids$ it holds that $t_k\geq \bmin$. Recall that from \cref{claim:first_bid} we know that $\bmin=\min \left\{t_i,\frac{\bsec}{c}\right\}$, so if $t_i\leq \bsec/c$ we immediately get the desired inequality $t_k\geq t_i$. It remains to check the case $\frac{\bsec}{c}=\bmin\leq t_k< t_i$. It holds that
\begin{align}
C_i(\bids)&= \left(1-\displaystyle\sum_{k\in N\setminus N_{\min}}\frac{\bmin}{\lrg\cdot\bidk}\right) \frac{\max\{\bmin,t_i\}}{n_{\min}}\notag\\
&= \left(1-\displaystyle\sum_{k\in N\setminus N_{\min}}\frac{\bmin}{\lrg\cdot\bidk}\right) t_i\notag\\
&> \left(1-\frac{n-1}{\lrg}\right) t_i\notag\\
&>  \frac{1}{\lrg}t_i \label{eq:aux1}
\end{align}
where the last inequality holds since $\lrg> 2(n-1)\geq n$ by assumption. Note
that if $t_i< \bsec$ then the right hand side of \eqref{eq:aux1} is equal to the
cost the machine would incur by declaring her true cost, i.e., $C_i(t_i,\bids_{-i}|\vecc t)$. Otherwise, if $t_i\geq \bsec$ then we get from \eqref{eq:aux1} and for $\bidi'\geq \max\{t_i,c\cdot\bsec\}$ that
\begin{equation*}
C_i(\bids)> \frac{1}{\lrg}t_i
\geq\frac{\bsec}{\lrg\cdot\bidi'}\bidi'
=\frac{\bsec}{\lrg\cdot\bidi'}\max\{\bidi', t_i\}
\geq C_i(\bidi',\bids_{-i}).
\end{equation*}
In each of the cases examined above we conclude that a beneficial deviation exists for the corresponding machine which is a contradiction to the original assumption that the starting configuration was at equilibrium.
\end{proof}

\begin{proof}[Proof of \cref{th:PoAnmachines}]
We now use the above claims to bound the makespan at any equilibrium $\bids$ of the game. We denote by $t_{\min}$ the minimum true execution cost of any machine (note that from \cref{claim:t(b)} we know that the machine who makes the minimum declaration $\bmin$ indeed achieves the minimum true cost $t_{\min}$). Then, the expected makespan is
\begin{align*}
\mathcal{M}(\bids)&= \left(1-\displaystyle\sum_{k\in N\setminus N_{\min}}\frac{\bmin}{\lrg\cdot\bidk}\right) \frac{\max\{\bmin,t_{\min}\}}{n_{\min}}n_{\min}+\displaystyle\sum_{k\in N\setminus N_{\min}}\frac{\bmin}{\lrg\cdot\bidk}\max\{\bidk,t_k\}\\
&\leq \max\{\bmin,t_{\min}\}+\frac{\bmin}{\lrg}(n-1)\\
&\leq t_{\min}+\frac{t_{\min}(n-1)}{\lrg}\\
&\leq \left(1+\frac{n-1}{\lrg}\right)t_{\min},
\end{align*}
where $\bmin$ is the optimal makespan. For the first equality we use \cref
{claim:sec_bid}, for the first inequality we use \cref{claim:bid>=true,claim:n_b}, and for the second inequality we use \cref{claim:first_bid,claim:t(b)}.

It can be easily verified that there exists at least one equilibrium $\bids$, e.g.\ choose $\bid{i^*}=t_{i^*}=t_{\min}$ for a machine $i^*\in\argmin_i t_i$ that has a minimum \emph{true} execution cost and  $\bidk=\max\{\lrg c \cdot t_{\min}, t_k\}$ for all other machines $k\in N\setminus \ssets{i}$.
\end{proof}

\section{Lower Bounds}

\subsection{Task-Independent Algorithms}
\label{append:lower_bound_taskind}

\begin{theorem} \label{th:lower_bound_taskind} No task-independent algorithm for
scheduling without payments on $n$  machines can have a Price of Anarchy better
than $\frac{1}{2}\sqrt{n}-o(1)$.
\end{theorem}

\begin{proof}
The bad instance here is simpler than the one in the proof of \cref
{th:lower_bound_anonymous_taskind}: consider $m=n$ tasks, for all of which
machine $i=1$ has true execution cost $1$ and all remaining machines are
way slower with costs $M=\sqrt{n}$.

We can assume that at every single-task
step there has to be
an equilibrium where the fast machine $i=1$ gets the job with probability at
least $p=\frac{1}{2}$, otherwise the expected makespan of that
single-shot algorithm would be at
least $ p\cdot 1+(1-p)\cdot M=\frac{1}{2}\sqrt{n}+\frac{1}{2}$,
resulting indeed in the desired Price of Anarchy lower bound since the optimal
cost is $1$ (by allocating the task to machine $i=1$).

Therefore, like in the proof of \cref
{th:lower_bound_anonymous_taskind}, it is not difficult to deploy a Chernoff
bound and see that the probability of machine $i=1$ receiving a $p (1-o
(1))=\frac{1}{2}-o(1)$ fraction of the overall number of jobs is at least $1-o
(1)$, thus
resulting in an expected makespan of at least $\frac{n}{2}-o(1)$. The optimal
allocation has a cost of $M=\sqrt{n}$, simply allocating each task $j$ to
machine $j$. This gives indeed a Price of Anarchy lower bound of $\frac{\frac
{n}{2}-o(1)}
{\sqrt{n}}=\frac{1}{2}\sqrt{n}-o(1)$.
\end{proof}

\subsection{The Mechanism of \texorpdfstring{\citet{K14}}{Koutsoupias}}
\label{sec:tight_lower_K14}
Consider an instance of $n$ machines and $n$ tasks, machine $i$ having (true) execution cost of $1$ for task $i$ and $M>1$ (to be determined later) for all other tasks. Formally,
$$
t_{i,j}=
\begin{cases}
1, &j=i,\\
M, &j\neq i.
\end{cases}
$$
Consider running the (truthful) mechanism of \citet{K14} independently for each task. Then (see \citep[Eq.~(1)]{K14}) the probability of a specific task $j$ getting assigned to its unique ``fast'' machine (that is machine $i=j$ having execution cost $1$) can be computed as
$$
p=\int_{0}^1\left(1-\frac{y}{M}\right)^{n-1}\,dy=\frac{M}{n}\left[1-\left(1-\frac{1}{M}\right)^n \right].
$$
Thus, the probability of at least one task being allocated to a ``slow'' machine is $1-p^n$, since all tasks are being assigned independently. At such an event, the resulting makespan would be at least $M$ (there would be at least one machine executing a task of duration $M$). So, the total expected makespan is at least $(1-p^n)M$, which by selecting $M=n^3$ becomes
$$
(1-p^n)M=n^3-n^{2n+3}\left[1-\left(1-\frac{1}{n^3}\right)^n \right]^n
$$
and it is a matter of straightforward calculations to verify that for any number of machines $n\geq 2$ this is lower bounded by
$$
(1-p^n)M\geq \frac{n(n+1)}{2+\frac{3}{n}}=\frac{n(n+1)}{2+o(1)}.
$$
Taking into consideration that the optimal makespan is $1$ (by allocating each task $j$ to machine $j$), this gives indeed a lower bound of $\frac{n(n+1)}{2+o(1)}$ to the expected makespan that asymptotically matches the mechanism's upper bound of $\frac{n(n+1)}{2}$ given in \citep{K14}.

\end{document}